\documentclass[11pt,a4paper,reqno,oneside]{amsart}

\usepackage{amsaddr}

\usepackage{amsmath} 
\usepackage{amsthm} 
\usepackage{amssymb} 
\usepackage{bbm} 

\usepackage{emptypage}
\usepackage{enumitem} 
\usepackage{mathtools} 
\usepackage{mathrsfs} 

\usepackage{hyperref}
\usepackage{cleveref,cite}

\newcommand{\rmi}{\mathrm{i}}
\newcommand{\RR}{\mathbb{R}}
\newcommand{\CC}{\mathbb{C}}

\usepackage{xcolor}

\newtheorem{theorem}{Theorem}[section]
\newtheorem*{theorem*}{Theorem}

\newtheorem{lemma}[theorem]{Lemma}
\newtheorem{corollary}[theorem]{Corollary}

\theoremstyle{definition}

\theoremstyle{remark}
\newtheorem{remark}[theorem]{Remark}
\newtheorem*{remark*}{Remark}

\newcommand{\con}[1]{\mathbb{#1}}
\newcommand{\C}{\con{C}} 
\newcommand{\R}{\con{R}} 
\renewcommand{\H}{\mathcal{H}}

\newcommand{\cL}{\mathcal{L}}

\newcommand{\half}{\frac{1}{2}}
\DeclareMathOperator{\Dom}{Dom}
\DeclareMathOperator{\sgn}{sgn}

\numberwithin{equation}{section}
\title[Congruence of Dirac operators]{Congruence of Dirac operators with applications to~generalized MIT bag models}

\author[J. Duran]{Joaquim Duran}
\address{Centre de Recerca Matem\`atica, Edifici C, Campus Bellaterra, 08193 Bellaterra, Spain}
\email{jduran@crm.cat}

\author[K. Pankrashkin]{Konstantin Pankrashkin}\thanks{Corresponding author: K. Pankrashkin}
\address{Carl von Ossietzky Universit\"at Oldenburg, Fakult\"at V -- Mathematik und Naturwissenschaften, Institut f\"ur Mathematik, Ammerl\"ander Heerstr. 114--118, 26129 Oldenburg, Germany}
\email{konstantin.pankrashkin@uol.de}


\begin{document}

\begin{abstract}
    We highlight a simple congruence transform that shifts coupling parameters for Dirac operators with shell interactions. As one of the consequences, this leads to new observations concerning the self-adjointness and the infinite mass interpretation of generalized MIT bag models.
\end{abstract}

\maketitle

\section{Introduction}

The goal of the present note is to highlight a simple congruence transform (see Section~\ref{sec:mainresults}) that leads to new observations related to Dirac operators with $\delta$-shell interactions, which attracted considerable attention during the last decade; see \cite{cas,vega,graz}. While the transform itself implicitly appeared in some previous papers (see~\cite[Section~2]{Pizzichillo2021} and~\cite[Lemma~4.5]{converge}), its overall relevance for the study of Dirac operators was not discussed so far to our knowledge.

Although most literature on shell interactions deals with the two- and three-dimensional cases, in this note we consider arbitrary dimensions in order to emphasize that the congruence transform only depends on the algebraic structure of the operators, which we introduce next. 

If $n\geq 2$ and $N:=2^{\lfloor\frac{n+1}{2}\rfloor}$, then there exist
$N\times N$ Hermitian and pairwise anticommuting matrices $\alpha_1,\dots, \alpha_n, \beta$ 
with
\begin{equation*}
    \alpha_k^2 = I_N = \beta^2, \quad \text{where } I_N \text{ is the } N\times N \text{ identity matrix}.
\end{equation*}
The associated Dirac operator $D_m$ with mass $m\in\RR$ is the differential operator acting on vector functions $f:\RR^n\to\CC^N$ by
\[
D_mf:=-\rmi\sum_{k=1}^n\alpha_k\partial_k f +m\beta f.
\]
Throughout the work, $\Sigma\subset \R^n$ will be a Lipschitz hypersurface that splits~$\R^n$ into open sets $\Omega_+$ and $\Omega_-$, so that~$\Sigma=\partial\Omega_+\cap\partial\Omega_-$ and $\RR^n=\Omega_+\cup\Sigma\cup\Omega_-$, and we denote by $\nu$ the unit normal vector field at~$\Sigma$ pointing outwards $\Omega_+$. Given a function $f\in L^2(\R^n,\C^N)$, we denote by $f_{\pm}$ its restriction to $\Omega_\pm$.

For $s_\pm\in[0,1]$ consider the Dirac-Sobolev spaces
\[
H^{s_\pm}_\alpha(\Omega_\pm,\CC^N):=\big\{f_\pm\in H^{s_\pm}(\Omega_\pm,\CC^N):\ D_0f_\pm\in L^2(\Omega_\pm,\CC^N)\big\}
\]
equipped with the norm given by $\|\cdot\|^2_{H^{s_\pm}_\alpha(\Omega_\pm,\CC^N)}:=\|\cdot\|^2_{H^{s_\pm}(\Omega_\pm,\CC^N)}+\|D_0\cdot\|^2_{L^2(\Omega_\pm,\CC^N)}$. Recall that, due to the general theory of Sobolev spaces, under suitable assumptions on~$\Sigma$ (for example, if $\Sigma$ is bounded or has a nice behavior at infinity) the restriction maps
\[
\gamma_\pm: \, C^\infty_c(\overline \Omega_\pm,\CC^N)\ni f\mapsto f|_\Sigma \in L^2(\Sigma,\CC^N)
\]
uniquely extend to bounded trace maps from $H_\alpha^{s_\pm}(\Omega_\pm,\CC^N)$ to $H^{s_\pm-\half}(\Sigma,\CC^N)$ for a suitable range of $s_\pm$. A proof for~$n\in\{2,3\}$ and  $s_\pm \in [\frac{1}{2},1]$ for the case of compact Lipschitz $\Sigma$ can be found in~\cite[Lemma~4.1]{graz},  and \cite[Proposition~2.1]{O-BV} established the required property for $s_\pm = 0$ and compact smooth $\Sigma$ for $n=3$. The same techniques apply for arbitrary dimensions $n$.

We are now ready to introduce the so-called Dirac operator with electrostatic and Lorentz scalar~$\delta$-shell interactions supported on $\Sigma$. Given coupling constants $\eta,\tau\in\RR$ and regularity parameters $s_\pm\in[0,1]$, it is the linear operator in~$L^2(\RR^n,\CC^N)$ acting by
\[
A^{s_+,s_-}_{m,\eta,\tau} :\ f\mapsto D_m f \quad \text{in the distributional sense in } \RR^n\setminus\Sigma
\]
on the domain
\begin{align*}
\Dom A^{s_+,s_-}_{m,\eta,\tau}:=\Big\{
 f\in & L^2(\RR^n,\CC^N):\ f_\pm \in H^{s_\pm}_\alpha(\Omega_\pm,\CC^N),\\
 &\rmi(\alpha\cdot\nu) \big( \gamma_+f_+ - \gamma_-f_-\big) + \tfrac{1}{2} (\eta I_N + \tau \beta ) \big( \gamma_+f_+ + \gamma_- f_-\big) = 0 \text{ on }\Sigma
\Big\}.
\end{align*}
The condition  for the traces $\gamma_\pm f_\pm$ is often referred to as a transmission condition, and the operator $A^{s_+,s_-}_{m,\eta,\tau}$ is often formally written as 
$A^{s_+,s_-}_{m,\eta,\tau}=D_m+(\eta I_N + \tau \beta )\delta_\Sigma$.
The parameters~$s_\pm$ are chosen to obtain a self-adjoint operator in $L^2(\RR^n,\CC^N)$, and the choice is known to depend on the regularity of $\Sigma$; see e.g.~\cite{graz,bpz,Pizzichillo2021}. Our observations in Section~\ref{sec:mainresults} show that all the operators of this type sharing the same value $\eta^2-\tau^2 \neq 0$ are related to each other by a simple congruence transform and, as a result, they share the same Sobolev regularity of the self-adjointness domains. In Remark~\ref{rmk:dimension} we treat the case $\eta^2-\tau^2 = 0$, for which the existence of congruence transforms seems to depend on the dimension $n$.

In the specific case $\eta^2-\tau^2=-4$ it is easily seen that the operator decouples into a direct sum of two operators in $L^2(\Omega_\pm,\CC^N)$, known as generalized MIT bag operators, which include the celebrated MIT bag model; see \cite{eigcurves}. Our observations in Section~\ref{sec:applications} show that all the generalized MIT bag operators are congruently equivalent to each other. As a result, we show for the first time the $H^1$-regularity of their operator domains for convex $\Omega_+$, which follows by a simple application of recent results obtained in~\cite{Pankrashkin2025} on the MIT bag model; see Section~\ref{sec:selfadj}. Furthermore, our considerations lead to the first infinite mass interpretation of generalized MIT bag boundary conditions, in the spirit of the interpretation for the MIT bag model studied in \cite{ALTMR,barb,bbz,SV}; see Section~\ref{sec:infmass}.

Before going into the details, in the next section we present some elementary results in functional analysis that will be used throughout the note.

\section{Elementary lemmas in functional analysis} \label{sec:FA}

Let $\mathcal H$ be a Hilbert space, and denote by $\cL(\H)$ the space of bounded linear operators in $\mathcal H$.
Recall that for a linear operator $A$ and a bounded linear operator $B$ in a Hilbert space $\H$, by~$B^*AB$ one means the operator in $\H$
acting as indicated on the domain
\[
\Dom B^*AB=\big\{x\in \H:\ Bx\in\Dom A\big\}.
\]
Noting that $B$ is in general not a unitary operator, in this section we discuss the transferring of some properties of $A$ to the congruently equivalent operator $B^*AB$. The first property is self-adjointness.

\begin{lemma}\label{lemma2}
	Let $A$ be a self-adjoint operator in a Hilbert space $\H$. If $B\in\cL(\H)$ is bijective, then the operator $B^*AB$ is self-adjoint.
\end{lemma}

Although this result follows from elementary properties of the adjoint operator (see, for example, \cite[Prop.~1.7~$(ii)$]{Schmudgen2012}), for the sake of completeness we give here a short self-contained proof.

\begin{proof}[Proof of Lemma~\ref{lemma2}]
	By assumption $B^{-1}\in\cL(\H)$, hence $B^*AB$ is densely defined. For $(x,y)\in\H\times\H$,
	the conditions $x\in \Dom(B^*AB)^*$ and $y:=(B^*AB)^*x$ are equivalent to
	\[
	\langle x,B^*AB  z\rangle=\langle y,z\rangle \text{ for all }z\in \Dom B^*AB.
	\]
	We have $\langle x,B^*AB  z\rangle=\langle Bx,A Bz\rangle$, and by denoting $w:=Bz$ we rewrite the last condition as
	\[
	\langle Bx,Aw\rangle=\langle y,B^{-1}w\rangle\equiv \langle (B^{-1})^*y,w\rangle  \text{ for all }w\in \Dom A,
	\]
	which, by the self-adjointness of $A$, means that $Bx\in \Dom A^*\equiv \Dom A$ with $ABx \equiv A^*(Bx) =(B^{-1})^*y\equiv (B^*)^{-1}y$. In conclusion, $x\in \Dom B^*AB$ with $B^*ABx=y=(B^*AB)^*x$, as claimed.
\end{proof}

The second property is the convergence behavior.

\begin{lemma}\label{lem11}
	Let $\H$ be a Hilbert space and $\H_0\subset\H$ be a closed subspace. Let $A_n$ be self-adjoint operators in $\H$
	and $A$ be a self-adjoint operator in $\H_0$. Further let $B,C\in \cL(\H)$ be self-adjoint such that $B$ is bijective
	and $\H_0$ is an invariant subspace for both $B$ and $C$. If for some $\lambda\in\CC\setminus\RR$ one has
	the norm convergence
	\begin{equation}
		\label{anl}
		(A_n-\lambda)^{-1}\xrightarrow{n\to\infty}(A-\lambda)^{-1}\oplus 0,
	\end{equation}
	with respect to the decomposition $\H=\H_0\oplus \H_0^\perp$,
	then for all $z\in\CC\setminus\RR$ one has the norm convergence
	\[
	(BA_nB+C-z)^{-1}\xrightarrow{n\to\infty}(BAB+C-z)^{-1}\oplus 0.
	\]	
\end{lemma}

\begin{proof}
	Denote $	K:=B^{-1}CB^{-1}-z B^{-2}$ 
	and observe that $K\in \cL(\H)$ with $K(\H_0)\subset \H_0$. To have a simpler writing we denote $B:\H_0\to\H_0$ still by $B$, analogously for $C$ and $K$. 	 Since
	\begin{align*}
		(BA_nB+C-z)^{-1}&=B^{-1}(A_n+K)^{-1}B^{-1},\\
		(BAB+C-z)^{-1}\oplus 0&=B^{-1}\big( (A+K)^{-1}\oplus 0\big)B^{-1},
	\end{align*}
	it is sufficient to show the norm convergence $ (A_n+K)^{-1}\xrightarrow{n\to\infty}(A+K)^{-1}\oplus 0$.
	To this end, notice that $(A_n-\lambda)^{-1}-(A_n+K)^{-1}=(A_n-\lambda)^{-1}(K+\lambda)(A_n+K)^{-1}$, namely
	\[
	(A_n-\lambda)^{-1}=\big(1+(A_n-\lambda)^{-1}(K+\lambda)\big)(A_n+K)^{-1}.
	\]
	This identity shows that the operators $T_n:=1+(A_n-\lambda)^{-1}(K+\lambda)$ are bounded with bounded inverse, and $(A_n+K)^{-1}=T_n^{-1}(A_n-\lambda)^{-1}$.
	Similarly, $(A+K)^{-1}=T^{-1}(A-\lambda)^{-1}$, with the operator
	$T:=1+(A-\lambda)^{-1}(K+\lambda)\in \cL(\H_0)$, and $T^{-1}\in \cL(\H_0)$.
	The convergence \eqref{anl} leads to $T_n\xrightarrow{n\to\infty} 1+\big((A-\lambda)^{-1}\oplus 0\big)(K+\lambda)\equiv T\oplus 1$.
	Since the limit operator $T\oplus 1$ is invertible, it follows that $T_n^{-1}\xrightarrow{n\to\infty} (T\oplus 1)^{-1}\equiv T^{-1}\oplus 1$,
	and then
	\begin{align*}
		(A_n+K)^{-1}&=T_n^{-1}(A_n-\lambda)^{-1}\xrightarrow{n\to\infty} (T^{-1}\oplus 1)\Big((A-\lambda)^{-1}\oplus 0\Big)\\
		&=\big(T^{-1}(A-\lambda)^{-1}\big)\oplus 0\equiv (A+K)^{-1}\oplus 0. \qedhere
	\end{align*}
\end{proof}

\section{The congruence transforms} \label{sec:mainresults}

Given $t\in \RR$, consider the $N\times N$ Hermitian matrix
\[
B_t:=\Big(\cosh \frac{t}{2}\Big)I_N+\Big(\sinh \frac{t}{2}\Big)\beta\equiv \exp \left( \frac{t}{2} \beta \right),
\]
which clearly satisfies $B_t B_r=B_{t+r}$, $B_0=I_N$, $(B_t)^{-1}=B_{-t}$ for $t,r\in\RR$. The following is the key observation of this note, which relates massless Dirac operators with shell interactions parameterized by pairs $(\eta,\tau)$ and $(\eta_0,\tau_0)$ lying in a common hyperbola, and which allows to apply the abstract results of Section~\ref{sec:FA}.

\begin{theorem} \label{thm:conjugations}
	Let $\eta,\tau,\eta_0,\tau_0\in\RR$ be such that $\eta^2-\tau^2=\eta_0^2-\tau_0^2\ne 0$.
	Then there exists $r\in\RR$ such that
\[
	\text{either} \quad \text{(i) }\, 
	A^{s_+,s_-}_{0,\eta,\tau}=B_r A^{s_+,s_-}_{0,\eta_0,\tau_0}B_r
	\qquad \text{or}\quad \text{(ii) } \,
	A^{s_+,s_-}_{0,\eta,\tau}=-\beta B_r A^{s_+,s_-}_{0,\eta_0,\tau_0} \beta B_r.
\]
\end{theorem}

As an immediate consequence we get the following criterion for self-adjointness for the (not necessarily massless) Dirac operators with electrostatic and Lorentz scalar $\delta$-shell interactions.

\begin{corollary} \label{cor:selfa}
	Let $m_0,\eta_0,\tau_0\in\RR$ with $\eta_0^2-\tau_0^2\ne 0$ and $s_+,s_-\in[0,1]$ be such that $A^{s_+,s_-}_{m_0,\eta_0,\tau_0}$ is self-adjoint. Then all the operators $A^{s_+,s_-}_{m,\eta,\tau}$ 
	with $m\in\RR$ and $\eta^2-\tau^2=\eta_0^2-\tau_0^2$ 
	are self-adjoint.
\end{corollary}

\begin{proof}
	Since bounded symmetric perturbations don't affect the self-adjointness, it is sufficient to consider the case $m_0=m=0$. Then the claim follows by combining the representations~$(i)$ and~$(ii)$ of Theorem~\ref{thm:conjugations} with Lemma~\ref{lemma2}.
\end{proof}

The distinction in Theorem~\ref{thm:conjugations} between the representations $(i)$ and $(ii)$ is related to whether the pairs $(\eta,\tau)$ and $(\eta_0,\tau_0)$ lie in the same branch of the hyperbola or not. This is unveiled by the following two lemmas, which combined together lead to the proof of Theorem~\ref{thm:conjugations}.

\begin{lemma} \label{lemma1}
	Let $k\in\RR\setminus\{0\}$. For all $s_\pm\in[0,1]$ and all $t\in\RR$ it holds that
	\begin{align}
        A^{s_+,s_-}_{0,k\cosh t,k\sinh t}&	=B_t A^{s_+,s_-}_{0,k,0} B_t, \label{conj1} \\
        A^{s_+,s_-}_{0,k\sinh t,k\cosh t}&	=B_t A^{s_+,s_-}_{0,0,k} B_t. \label{conj2}
	\end{align}
\end{lemma}

\begin{lemma} \label{lemma0}
	For all $\eta,\tau\in\RR$ and all $s_\pm\in[0,1]$ it holds that $A^{s_+,s_-}_{0,\eta,\tau}=-\beta A^{s_+,s_-}_{0,-\eta,-\tau}\beta$.
\end{lemma}

Assuming for a moment Lemmas~\ref{lemma1} and \ref{lemma0} (which we prove next), we show the proof of our key result.

\begin{proof}[Proof of Theorem~\ref{thm:conjugations}]
Let $K:=\eta_0^2-\tau_0^2\equiv \eta^2-\tau^2$.

\underline{Case A:} $K>0$.

$\bullet$ Subcase A.1: $\sgn \eta=\sgn\eta_0$. Then one can find
$k\in\RR$, with $k^2=K$, and $t_0,t\in\RR$ such that
\[
(\eta_0,\tau_0)=(k\cosh t_0,k\sinh t_0) \quad \text{and} \quad (\eta,\tau)=(k\cosh t,k\sinh t).
\]
By \eqref{conj1} one has $A^{s_+,s_-}_{0,\eta_0,\tau_0}	=B_{t_0} A^{s_+,s_-}_{0,k,0} B_{t_0}$ and $A^{s_+,s_-}_{0,\eta,\tau}	=B_{t} A^{s_+,s_-}_{0,k,0} B_{t}$.
Hence,
\begin{equation}
	\label{conj3}
A^{s_+,s_-}_{0,\eta,\tau}=B_t (B_{t_0})^{-1}A^{s_+,s_-}_{0,\eta_0,\tau_0}(B_{t_0})^{-1}B_{t}=
B_{t-t_0}A^{s_+,s_-}_{0,\eta_0,\tau_0}B_{t-t_0},
\end{equation}
which is the representation $(i)$ with $r:=t-t_0$.

$\bullet$ Subcase A.2: $\sgn \eta=-\sgn\eta_0$. Then $A^{s_+,s_-}_{0,-\eta,-\tau}=B_r A^{s_+,s_-}_{0,\eta_0,\tau_0} B_r$
with some $r\in\RR$ due to Subcase~A.1, while $A^{s_+,s_-}_{0,\eta,\tau}=-\beta A^{s_+,s_-}_{0,-\eta,-\tau} \beta$
due to Lemma~\ref{lemma0}. Since $\beta$ and $B_r$ commute, one arrives at the representation $(ii)$.

\underline{Case B:} $K<0$. 

$\bullet$ Subcase B.1: $\sgn \tau=\sgn\tau_0$. Then one can find
$k\in\RR$, with $k^2=-K$,  and $t_0,t\in\RR$ such that
\[
(\eta_0,\tau_0)=(k\sinh t_0,k\cosh t_0),\quad (\eta,\tau)=(k\sinh t,k\cosh t).
\]
By \eqref{conj2} one has $A^{s_+,s_-}_{0,\eta_0,\tau_0}	=B_{t_0} A^{s_+,s_-}_{0,0,k} B_{t_0}$
and $A^{s_+,s_-}_{0,\eta,\tau}	 =B_{t} A^{s_+,s_-}_{0,0,k} B_{t}$,
which again leads to the identity \eqref{conj3} and yields the representation $(i)$ with $r:= t-t_0$.

$\bullet$ Subcase B.2:  $\sgn \tau= -\sgn\tau_0$. We apply Subcase~B.1 to $A^{s_+,s_-}_{0,-\eta,-\tau}$
and then use $A^{s_+,s_-}_{0,\eta,\tau}=-\beta A^{s_+,s_-}_{0,-\eta,-\tau} \beta$ to arrive
at the representation $(ii)$.
\end{proof}

We are only left with proving Lemmas~\ref{lemma1} and~\ref{lemma0}.

\begin{proof}[Proof of Lemma~\ref{lemma1}]
First, note that the action of all Dirac operators in \eqref{conj1} and \eqref{conj2} corresponds to $D_0$, which does not contain $\beta$. Due to the anti-commutation of $\alpha_j$ and $\beta$, for all $x\in \RR^n$ we have
$(\alpha\cdot x)B_t=B_{-t}(\alpha\cdot x)$. Hence, for any $f\in H^0_\alpha(\Omega_\pm,\CC^N)$ one has $B_t D_0 B_t f= B_t B_{-t} D_0 f=D_0 f$.
It is also clear that $(B_t f)_\pm\in H^{s_\pm}(\Omega_\pm,\CC^N)$ is equivalent to $f_\pm\in H^{s_\pm}(\Omega_\pm,\CC^N)$.
Hence, in order to prove \eqref{conj1} and \eqref{conj2} one simply needs to verify that:
\begin{itemize}
	\item[(a)] A function $f$ satisfies the transmission condition for $A^{s_+,s_-}_{0,k\cosh t,k\sinh t}$ if and only if 
	the function $B_t f$ satisfies the transmission condition for $A^{s_+,s_-}_{0,k,0}$.
	\item[(b)] A function $f$ satisfies the transmission condition for $A^{s_+,s_-}_{0,k\sinh t,k\cosh t}$ if and only if 
	the function $B_t f$ satisfies the transmission condition for $A^{s_+,s_-}_{0,0,k}$.
\end{itemize}

We first address the proof of (a). Note that the transmission condition
\[
\rmi(\alpha\cdot\nu) \big( \gamma_+f_+ - \gamma_-f_- \big) + \frac{1}{2} k\big((\cosh t) I_N + (\sinh t) \beta \big) \big( \gamma_+f_+ + \gamma_- f_-\big) = 0
\]
for $A^{s_+,s_-}_{0,k\cosh t,k\sinh t}$ can be rewritten as
\[
\rmi(\alpha\cdot\nu) \big( \gamma_+f_+ - \gamma_-f_- \big) + \frac{1}{2} k B_{2t} \big( \gamma_+f_+ + \gamma_- f_-\big) = 0
\]
and, after multiplying by $B_{-t}$ and noting that $B_{-t}(\alpha\cdot\nu)=(\alpha\cdot \nu)B_t$, it takes the equivalent form
\[
\rmi(\alpha\cdot\nu) \big( \gamma_+(B_tf)_+ - \gamma_-(B_tf)_- \big) + \frac{1}{2}  (kI_N+0\beta)\big( \gamma_+ (B_tf)_+ + \gamma_- (B_tf)_-\big) = 0,
\]
which is exactly the transmission condition for $A^{s_+,s_-}_{0,k,0}$ to be satisfied by $B_tf$.

We conclude with the proof of (b). In a similar fashion, we rewrite the transmission condition
\[
\rmi(\alpha\cdot\nu) \big( \gamma_+f_+ - \gamma_-f_- \big) + \frac{1}{2} k\big((\sinh t) I_N + (\cosh t) \beta \big) \big( \gamma_+f_+ + \gamma_- f_-\big) = 0
\]	
for $A^{s_+,s_-}_{0,k\sinh t,k\cosh t}$ as
\[
\rmi(\alpha\cdot\nu) \big( \gamma_+f_+ - \gamma_-f_- \big) + \frac{1}{2} k \beta B_{2t} \big( \gamma_+f_+ + \gamma_- f_-\big) = 0,
\]
and the multiplication by $B_{-t}$ results in
\[
\rmi(\alpha\cdot\nu) \big( \gamma_+(B_tf)_+ - \gamma_-(B_tf)_-\big) + \frac{1}{2}  (0 I_N+k\beta)\big( \gamma_+ (B_tf)_+ + \gamma_- (B_tf)_-\big) = 0,
\]
which is the transmission condition for $A^{s_+,s_-}_{0,0,k}$ to be satisfied by $B_tf$.
\end{proof}

\begin{proof}[Proof of Lemma~\ref{lemma0}]
The action of the Dirac operators $A^{s_+,s_-}_{0,\pm \eta,\pm\tau}$  corresponds to $D_0$, which does not contain $\beta$.
 Due to the anti-commutation of $\alpha_j$ and $\beta$, for all $x\in \RR^n$ one has~$(\alpha\cdot x)\beta=-\beta(\alpha\cdot x)$. For all $f\in H^0_\alpha(\Omega_\pm,\CC^N)$ one has then $\beta D_0 \beta f= \beta (-\beta D_0 f)=-D_0 f$, and it is obvious that $(\beta f)_\pm\in H^{s_\pm}(\Omega_\pm,\CC^N)$ is equivalent to $f_\pm\in H^{s_\pm}(\Omega_\pm,\CC^N)$.
 It remains to show that a function $f$ satisfies the transmission condition for $A^{s_+,s_-}_{0,\eta,\tau}$
if and only if the function~$\beta f$ satisfies the transmission condition for $A^{s_+,s_-}_{0,-\eta,-\tau}$.
The former takes the form
\[
\rmi(\alpha\cdot\nu) \big( \gamma_+f_+ - \gamma_-f_- \big) + \frac{1}{2} (\eta I_N + \tau \beta ) \big( \gamma_+f_+ + \gamma_- f_-\big) = 0.
\]
The multiplication by $\beta$ together with the anti-commutation in the first summand yields
\[
\rmi(\alpha\cdot\nu) \big( \gamma_+\beta f_+ - \gamma_-\beta f_- \big) - \frac{1}{2} (\eta I_N + \tau \beta ) \big( \gamma_+\beta f_+ + \gamma_- \beta f_-\big) = 0,
\]
which is the transmission condition for $A^{s_+,s_-}_{0,-\eta,-\tau}$ to be satisfied by $\beta f$.
\end{proof}

Similarly to the preceding computations one also shows:
\begin{lemma}\label{lem-zero}
	For any $\eta\in\RR$, $s_\pm\in[0,1]$ and $t\in\RR$ one has $B_t A_{0,\eta,\pm\eta}^{s_+,s_-}B_t=A_{0,e^{\pm t}\eta, \pm e^{\pm t}\eta}^{s_+,s_-}$.
\end{lemma}

\begin{proof}
	The action of both $A_{0,\eta,\pm\eta}^{s_+,s_-}$ and $A_{0,e^{\pm t}\eta, \pm e^{\pm t}\eta}^{s_+,s_-}$
	corresponds to $D_0$ and does not involve~$\beta$.  Hence, for any $f\in H^0_\alpha(\Omega_\pm,\CC^N)$ one has $B_t D_0 B_t f= B_t B_{-t} D_0 f=D_0 f$, and it is also clear that $(B_t f)_\pm\in H^{s_\pm}(\Omega_\pm,\CC^N)$ is equivalent to $f_\pm\in H^{s_\pm}(\Omega_\pm,\CC^N)$. It remains to check that a function $f$ satisfies the transmission condition for $A_{0,e^{\pm t}\eta, \pm e^{\pm t}\eta}^{s_+,s_-}$ if and only if the function $B_t f$ satisfies the transmission condition for $A^{s_+,s_-}_{0,\eta,\pm\eta}$. The latter means
	\[
	\rmi(\alpha\cdot\nu) B_t\big(\gamma_+f_+ - \gamma_-f_- \big) + \frac{1}{2} \eta(I_N \pm \beta ) B_t\big( \gamma_+f_+ + \gamma_- f_-\big).
	\]
	Due to $(\alpha\cdot\nu) B_t=B_{-t}(\alpha\cdot\nu)$, the multiplication by $B_t$ transforms it into
	\[
	\rmi(\alpha\cdot\nu) \big(\gamma_+f_+ - \gamma_-f_- \big) + \frac{1}{2} \eta (I_N \pm \beta ) B_{2t}\big( \gamma_+f_+ + \gamma_- f_-\big),
	\]
	and the easy computation
	\begin{align*}
		(I_N \pm \beta ) B_{2t}&=(I_N\pm\beta)\big( \cosh (t)I_N +\sinh(t)\beta\big)\\
		& =\big(\cosh(t)\pm\sinh(t)\big)(I_N\pm\beta)=e^{\pm t}(I_N\pm \beta)
	\end{align*}
	gives the equivalent form	
	\[
	\rmi(\alpha\cdot\nu) \big(\gamma_+f_+ - \gamma_-f_- \big) + \frac{1}{2} e^{\pm t}\eta (I_N \pm \beta ) \big( \gamma_+f_+ + \gamma_- f_-\big),
	\]
	which is exactly the transmission condition for $A_{0,e^{\pm t}\eta, \pm e^{\pm t}\eta}^{s_+,s_-}$.
\end{proof}

\begin{remark} \label{rmk:dimension}
When $\eta^2-\tau^2 = 0$, the situation seems a bit different. Note first that the case $(\eta,\tau)=(0,0)$ corresponds
to the free Dirac operator, and the associated transmission condition  $\gamma_+f_+-\gamma_-f_-=0$ is invariant under multiplication by constant matrices. As for $(\eta,\tau)\ne(0,0)$ 
the transmission condition allows for a non-trivial jump $\gamma_+ f_+-\gamma_- f_-$ along~$\Sigma$,
no pointwise congruence transform with a constant matrix can transform $A_{0,\eta,\pm\eta}^{s_+,s_-}$ with $\eta\ne 0$ into $\pm A_{0,0,0}^{s_+,s_-}$.

On the other hand, by combining Lemma~\ref{lem-zero} with Lemma~\ref{lemma0} one sees
that for each fixed $\rho\in\{-1,+1\}$ all operators $A_{0,\eta,\rho\eta}^{s_+,s_-}$ with $\eta\ne 0$ are congruently equivalent to each other (up to the sign). In other words, if one considers the punctured lines
\[
\ell_\pm:=\big\{(\eta,\pm\eta): \  \eta\ne 0\big\},
\]
then $A_{0,\eta,\tau}^{s_+,s_-}$ and $A_{0,\eta_0,\tau_0}^{s_+,s_-}$ with $\eta^2-\tau^2=\eta_0^2-\tau_0^2=0$
are congruently equivalent (up to the sign) provided both $(\eta,\tau)$ and $(\eta_0,\tau_0)$ are on the same $\ell_\pm$,
i.e. $\sgn(\eta\tau)=\sgn(\eta_0\tau_0)$.

Further, note that for any $\eta$ one can transform $A_{0,\eta,\eta}^{s_+,s_-}$ into $A_{0,\eta,-\eta}^{s_+,s_-}$
(and vice versa) by a congruence transform at least in odd dimensions (which then connects the punctured parametric lines $\ell_+$ and $\ell_-$).
Namely, when $n$ is odd, then the matrix
	\begin{equation} \label{eq:gamma}
		\gamma:=  \rmi^{n(n-1)/2}\alpha_1 \cdot \ldots \cdot\alpha_n,
	\end{equation}
	is Hermitian, commutes with all $\alpha_j$ and anticommutes with $\beta$, and a straightforward computation as in the proof of Lemma~\ref{lemma0} shows that
	\begin{equation} \label{eq:fail}
		A^{s_+,s_-}_{0,\eta,\eta} = \gamma A^{s_+,s_-}_{0,\eta,-\eta} \gamma.
	\end{equation}
	However, when $n$ is even, by Schur's representation lemma one sees that $\gamma$ as in \eqref{eq:gamma} is a scalar multiple of $\beta$, and thus \eqref{eq:fail} no longer holds true due to Lemma~\ref{lemma0}. We believe that there can be a dimensional obstruction to the congruent equivalence in even dimensions for $\eta^2-\tau^2=0$.
\end{remark}

	The above observations can be summarized as follows.

\begin{theorem}
Let $(\eta,\tau),(\eta_0,\tau_0)\in\RR^2\setminus\{(0,0)\}$ with $\eta^2-\tau^2=\eta_0^2-\tau_0^2=0$. If $n$ is even, assume additionally $\sgn(\eta\tau)=\sgn(\eta_0\tau_0)$. Then one can find an invertible matrix $\Gamma$
and $\rho\in\{-1,+1\}$ with $A^{s_+,s_-}_{0,\eta,\tau}=\rho\Gamma^* A^{s_+,s_-}_{0,\eta_0,\tau_0}\Gamma$.	In particular, the self-adjointness
of $A^{s_+,s_-}_{0,\eta,\tau}$ is equivalent to the self-adjointness of $A^{s_+,s_-}_{0,\eta_0,\tau_0}$.
\end{theorem}

\section{Applications to generalized MIT bag operators} \label{sec:applications}

In this section we consider in greater detail the case $\eta^2-\tau^2=-4$. In this case the transmission condition along $\Sigma$ is actually a confinement, namely, it decouples into two separate boundary conditions for $f_+$ and $f_-$, which are
\[ \label{bc1}
	\gamma_\pm f_\pm =\pm\dfrac{\rmi}{2}(\eta I_N-\tau\beta)(\alpha\cdot\nu)\gamma_\pm f_\pm.
\]
Since in this case one can parameterize $\eta=2\rho\sinh t$ and $\tau=2\rho\cosh t$ for some unique $t\in\RR$ and~$\rho\in\{-1,1\}$, the above boundary conditions take the form
\begin{align}
	\label{bc1a}
	\gamma_\pm f_\pm & = \pm \rmi \rho\big((\sinh t) I_N-(\cosh t)\beta\big)(\alpha\cdot\nu)\gamma_\pm f_\pm.
\end{align}
Remark that for $t=0$, namely for $\eta=0$ and $\tau=2\rho\in\{-2,2\}$, one obtains the boundary condition
\[ \label{bc1b}
	\gamma_\pm f_\pm = \mp \rmi \rho\beta(\alpha\cdot\nu)\gamma_\pm f_\pm,
\]
which is referred to as \emph{MIT bag} for $\rho=+1$ or \emph{anti-MIT bag} for $\rho=-1$. In turn, the boundary condition \eqref{bc1a} is often called a \emph{generalized MIT bag} boundary condition.

Due to the decoupling of the transmission condition we can write
\[
A^{s_+,s_-}_{m,\eta,\tau}=T^{s_+}_{m,\Omega_+,t,\rho}\oplus T^{s_-}_{m,\Omega_-,t,\rho},
\]
where $T^{s_\pm}_{m,\Omega_\pm,t,\rho}$ are the linear operators in $L^2(\Omega_\pm,\CC^N)$ acting as $f\mapsto D_m f$ on the domains
\begin{align*}
\Dom T^{s_\pm}_{m,\Omega_\pm,t,\rho} & =\big\{
f\in H^{s_\pm}_\alpha(\Omega_\pm,\CC^N): \text{ the boundary condition \eqref{bc1a} holds on } \Sigma
\big\},
\end{align*}
and are called \emph{generalized MIT bag operators}. Notice that by Lemma~\ref{lemma1} we can write 
\begin{align*}
T^{s_+}_{0,\Omega_+,t,\rho}\oplus T^{s_-}_{0,\Omega_-,t,\rho}&=A^{s_+,s_-}_{0,2\rho\sinh t,2\rho\cosh t}=B_t A^{s_+,s_-}_{0,0,2\rho}B_t = (B_tT^{s_+}_{0,\Omega_+,0,\rho}B_t)\oplus (B_t T^{s_-}_{0,\Omega_-,0,\rho}B_t),
\end{align*}
and therefore
\begin{equation}
	\label{conj5}
	T^{s_\pm}_{0,\Omega_\pm,t,\rho}=B_tT^{s_\pm}_{0,\Omega_\pm,0,\rho}B_t.
\end{equation}
Similarly, combining Lemmas~\ref{lemma1} and~\ref{lemma0} we have
\begin{equation}
	\label{conj6}
	T^{s_\pm}_{0,\Omega_\pm,t,\rho}=-\beta B_tT^{s_\pm}_{0,\Omega_\pm,0,-\rho}\beta B_t.
\end{equation}
The identities~\eqref{conj5} and~\eqref{conj6} allow applying the abstract results of Section~\ref{sec:FA} to the generalized MIT bag operators. This is what we do in the following sections.

\subsection{Self-adjointness} \label{sec:selfadj}

As an immediate consequence of Lemma~\ref{lemma2} we get the following criterion for self-adjointness for the generalized MIT bag operators.
\begin{lemma}\label{lem4}
	If \emph{some} generalized MIT bag Dirac operator $T^{s_\pm}_{m_0,\Omega_\pm,t_0,\rho_0}$ is self-adjoint in $L^2(\Omega_\pm,\CC^N)$, then all the generalized MIT bag operators $T^{s_\pm}_{m,\Omega_\pm,t,\rho}$ with $m\in\RR$, $t\in\RR$, and $\rho\in\{-1,1\}$ are also self-adjoint.
\end{lemma}

\begin{proof}
Since bounded symmetric perturbations don't affect the self-adjointness, it is sufficient to consider the case $m_0=m=0$. Due to \eqref{conj5} and \eqref{conj6}, for some $r\in\RR$ one has
\[
\text{either} \quad T^{s_\pm}_{0,\Omega_\pm,t,\rho}= B_r T^{s_\pm}_{0,\Omega_\pm,t_0,\rho_0} B_r
\qquad \text{or} \quad
T^{s_\pm}_{0,\Omega_\pm,t,\rho}= -\beta B_r T^{s_\pm}_{0,\Omega_\pm,t_0,\rho_0} \beta B_r,
\]
and then the claim follows by Lemma~\ref{lemma2}.
\end{proof}

The previous lemma can be shortened to the following general principle: if the MIT bag Dirac operator 
on some domain $\Omega\subset\RR^n$ is self-adjoint, then all generalized MIT bag operators in $\Omega$
are self-adjoint \emph{with the same Sobolev regularity of the operator domain}. This principle readily leads to the self-adjointness of generalized MIT bag models for all those cases in which the self-adjointness of the MIT bag operator is already known. It should be noted, however, that for a variety of cases the self-adjointness of Dirac operators with shell interactions on suitable domains has been proved without using the above congruence~\cite{behl,bhm,graz,vega,Benguria2017Self,O-BV}. In particular,
the self-adjointness of the generalized MIT bag models with $s=\frac{1}{2}$ for compact Lipschitz $\Sigma$
and $n\in\{2,3\}$ is proved in \cite[Corol.~5.8(iii)]{graz}, while compact smooth $\Sigma$
with $s=1$ (as well as suitably behaved non-compact smooth~$\Sigma$) for arbitrary $n$ are covered in \cite{rab}. Nevertheless,
we may consider the case of convex domains as a new application.

\begin{theorem}[Self-adjointness on $H^1$ in convex domains]\label{thm32}
	Let $\Omega_+\subset\RR^n$ be a bounded convex domain. Then the generalized MIT bag operators $T^{1}_{m,\Omega_+,t,\rho}$ are self-adjoint in $L^2(\Omega_+,\CC^N)$ for all~$m\in\RR$, $\rho\in\{-1,1\}$ and all $t\in\RR$.
\end{theorem}

\begin{proof}
	The operator $T^{1}_{0,\Omega_+,0,1}$ is self-adjoint in $L^2(\Omega_+,\CC^N)$ by \cite{Pankrashkin2025}. Therefore, the claim follows by Lemma~\ref{lem4}.
\end{proof}

\subsection{Infinite mass limit} \label{sec:infmass}

Given $M\in\RR$, consider the Dirac operator in $\RR^n$ with mass $m$ in $\Omega_+$ and mass $m+M$ in $\Omega_-$, which is the self-adjoint operator in $L^2(\RR^n,\CC^N)$ defined as
\begin{gather*}
	Z_{m,M,0,1}:\ f\mapsto D_m f +M\mathbbm{1}_{\Omega_-}\beta f, \qquad
	\Dom Z_{m,M,0,1}=H^1(\RR^n,\CC^N).
\end{gather*}
For a variety of situations it is known that when $M\to+\infty$ the operator $Z_{m,M,0,1}$ approximates (in a suitably defined sense) the Dirac operator $T^{s}_{m,\Omega_+,0,1}$ in $\Omega_+$ with mass $m$ and MIT bag boundary condition.
One of the possible versions is the following, where $0_{\Omega_-}$ is the zero operator in $\Omega_-$.

\begin{lemma}\label{prop23}
	Let $n\in\{2,3\}$. If $\Omega_+$ is a bounded smooth domain, then for all $m\in\RR$ and all~$\lambda\in\CC\setminus\RR$
	one has the norm convergence
	\[
	(Z_{m,M,0,1}-\lambda)^{-1}\xrightarrow{M\to+\infty}(T^{1}_{m,\Omega_+,0,1}-\lambda)^{-1}\oplus 0_{\Omega_-}.
	\]
\end{lemma}

We refer to \cite{barb} for a proof when $n=2$ and to \cite{bbz} when $n=3$; note that in \cite{barb} some unbounded smooth domains $\Omega_+\subset\RR^2$ are also allowed, and some extensions to non-smooth domains are discussed in the recent paper~\cite{bkmp}.

Due to Lemma~\ref{lem11} the convergence is still valid if we replace both operators in Lemma~\ref{prop23} by congruently equivalent ones, and this simple observation leads to interesting consequences. Namely, consider the self-adjoint operators
\begin{equation*}
	Z_{m,M,t,\rho}:\ f\mapsto D_m f + \rho M\mathbbm{1}_{\Omega_-}\big( (\sinh t) I_N + (\cosh t)\beta\big) f, \qquad \Dom Z_{m,M,t,\rho} =H^1(\RR^n,\CC^N),
\end{equation*}
then the following result holds.

\begin{lemma} \label{lemma:res_criterion}
	Assume that for some $s\in[0,1]$ the MIT bag operator $T^{s}_{0,\Omega_+,0,1}$ is self-adjoint in~$L^2(\Omega_+,\CC^N)$
	and that  the norm convergence
	\[
	(Z_{0,M,0,1}-\lambda)^{-1}\xrightarrow{M\to+\infty}(T^{s}_{0,\Omega_+,0,1}-\lambda)^{-1}\oplus 0_{\Omega_-}
	\]
	holds for some $\lambda\in\CC\setminus\RR$. Then, for all $m\in\RR$, $t\in\RR$ and $\rho\in\{-1,1\}$ the norm convergence
	\[
	(Z_{m,M,t,\rho} - \lambda )^{-1} \xrightarrow{M\to+\infty}(T^{s}_{m,\Omega_+,t,\rho}-\lambda)^{-1}\oplus 0_{\Omega_-}
	\]
	holds for all $\lambda\in\CC\setminus\RR$.
\end{lemma}

\begin{proof}
	Since by hypothesis $T^{s}_{0,\Omega_+,0,1}$ is self-adjoint in~$L^2(\Omega_+,\CC^N)$, by Lemma~\ref{lem4} the operator $T^{s}_{m,\Omega_+,t,\rho}$ is also self-adjoint in~$L^2(\Omega_+,\CC^N)$. Moreover, since
	\begin{align*}
		T^{s}_{m,\Omega_+,t,1}&=B_tT^{s}_{0,\Omega_+,0,1}B_t+m\beta, & T^{s}_{m,\Omega_+,t,-1}&= -\beta B_tT^{s}_{0,\Omega_+,0,1} \beta B_t+m\beta,\\
		Z_{m,M,t,1}&=B_t Z_{0,M,0,1}B_t+m\beta, & Z_{m,M,t,-1}&= -\beta B_t Z_{0,M,0,1} \beta B_t+m\beta,
	\end{align*}
	and $L^2(\Omega_+,\CC^N)\oplus 0$ is an invariant subspace for both $\beta$ and $B_t$, the claim follows by Lemma~\ref{lem11}.
\end{proof}

A straightforward combination of Lemmas~\ref{prop23} and \ref{lemma:res_criterion} leads to the following effective interpretation of the generalized MIT bag boundary conditions.

\begin{theorem}[Generalized MIT bag as an infinite mass limit]\label{thm-conv}
	Let $n\in\{2,3\}$. If $\Omega_+$ is a bounded smooth domain, then for all $m\in\RR$, $t\in\RR$,  $\rho\in\{-1,1\}$ and $\lambda\in\CC\setminus\RR$, one has the norm convergence
	\[
	(Z_{m,M,t,\rho}-\lambda)^{-1}\xrightarrow{M\to+\infty}(T^{1}_{m,\Omega_+,t,\rho}-\lambda)^{-1}\oplus 0_{\Omega_-}.
	\]
\end{theorem}

\section*{Funding information}

This article is based upon work from COST Action 24122 mSPACE, supported by COST (European Cooperation in Science and Technology), \url{www.cost.eu}.

J.~Duran was supported by the Catalan grant 2021-SGR-00087,
 by the Spanish grant PID2025-168310NB-I00 funded by MCIN/AEI/10.13039/501100011033, by ERDF ``A way of making Europe". His work was also supported by the Spanish State Research Agency, through the Severo Ochoa and Mar\'ia de Maeztu Program for Centers and Units of Excellence in R\&D (CEX2020-001084-M), and more specifically by the grant CEX2020-001084-M-20-1. He also acknowledges CERCA Programme/Generalitat de Catalunya for institutional support.

K.~Pankrashkin was partially supported by the Deutsche Forschungsgemeinschaft (DFG, German Research Foundation), project 491606144.

\end{document}